\documentclass{IEEEtran}
\usepackage{cite}
\usepackage{amsmath,amssymb,amsfonts,amsthm}
\usepackage{algorithmic}
\usepackage{mathtools}
\usepackage{graphicx}
\usepackage{textcomp}
\usepackage{multirow}
\usepackage{balance}
\usepackage[linesnumbered,ruled]{algorithm2e}
\usepackage[dvipsnames]{xcolor}
\def\BibTeX{{\rm B\kern-.05em{\sc i\kern-.025em b}\kern-.08em
    T\kern-.1667em\lower.7ex\hbox{E}\kern-.125emX}}

\newtheorem{theorem}{Theorem}
\newtheorem{lemma}{Lemma}
\newtheorem{corollary}{Corollary}
\newtheorem{definition}{Definition}

\newtheorem{example}{Example}
\newtheorem{remark}{Remark}

\DeclareMathOperator*{\argmax}{argmax}

\begin{document}
\title{Intelligent Players in a Fictitious Play Framework\thanks{B. Vundurthy and V. Gupta are with the Department of Electrical Engineering at the University of Notre Dame {\tt{pvundurt,vgupta2}@nd.edu}. A. Kanellopoulos and K. Vamvoudakis are with the School of Aerospace Engineering at the Georgia Institute of Technology {\tt{ariskan,kyriakos}@gatech.edu}. }}
\author{Bhaskar Vundurthy, 
	\IEEEmembership{Member, IEEE}, 
	Aris Kanellopoulos,
	\IEEEmembership{Student Member, IEEE},
\\	Vijay Gupta, 
	\IEEEmembership{Senior Member, IEEE},
	Kyriakos Vamvoudakis, 
	\IEEEmembership{Senior Member, IEEE}
\thanks{This work was supported in part by ARO under grant No.s W$911$NF-$19-1-0270$ and W911NF-19-1-0483, by ONR Minerva under grant No. N$00014-18-1-2160$, by AFOSR under grant No. FA9550-21-1-0231, by DARPA under grant No.  FA8750-20-2-0502, and by NSF under grant No.s CAREER CPS-$1851588$ and ECCS-2020246.}}
\maketitle

\begin{abstract}
Fictitious play is a popular learning algorithm in which players that utilize the history of actions played by the players and the knowledge of their own payoff matrix can converge to the Nash equilibrium under certain conditions on the game. We consider the presence of an {\em intelligent} player that has access to the entire payoff matrix for the game. We show that by not conforming to fictitious play, such a player can achieve a better payoff than the one at the Nash Equilibrium. This result can be viewed both as a fragility of the fictitious play algorithm to a strategic intelligent player and an indication that players should not throw away additional information they may have, as suggested by classical fictitious play.  
\end{abstract}


\section{Introduction}
\label{sec:introduction}
Learning algorithms (see, e.g.,~\cite{fudenberg1998theory,marden2018game,young2004strategic}) can be viewed as a mechanism for the agents to discover their solution strategies under a solution concept such as a Nash equilibrium. Since natural and universal learning algorithms cannot converge to Nash equilibria~\cite{hart2003uncoupled}, convergence guarantees for particular learning algorithms are limited to specific game structures. In this paper, we specifically focus on fictitious play~\cite{brown1951} as the learning algorithm used by players. In fictitious play, each player builds a model of what the strategy of the other players is based on the historical actions taken by them and plays a best response to it. Analyzing the class of games for which fictitious play and its variants converge to a Nash equilibrium continues to be a direction of active research.

Two features of this algorithm are worth pointing out. First, almost all convergence results for the fictitious play algorithm assume that all players are following this algorithm. For a setting of games among strategic players, this seems a strong assumption requiring some form of {\em cooperation} among otherwise non-cooperative players. A quick thought, e.g., reveals that in a two-player game, a strategic player can force her opponent into a Stackelberg equilibrium with herself as the leader (and potentially gain in payoff) by deviating from the trajectory suggested by fictitious play. The first question of interest to us is to identify the optimal payoff that a strategic player can achieve by exploiting the fact that all the other players follow fictitious play. We show that a payoff higher than the one in Stackelberg equilibrium is indeed achievable.

The second feature of fictitious play is that the players do not use any further information about the game other than their own utility payoffs for various strategy combinations. This is desirable in that players that possess limited and distributed information about the game can still discover the solution. However, it does raise the question if a player with more knowledge can obtain a better payoff for herself (or all the players) than the one at Nash equilibrium. Once again, we show that a player that knows the entire payoff matrix for all the players can indeed improve its own payoff and in some settings, the payoffs of all players by using that information.

We also note that the convergence results on the classes of games for which fictitious play is known to converge is much larger when only two players are involved. When more than two players are present, it has been shown that the Nash equilibrium need not possess an {\em absorption} property (where a strategy profile leaves no incentive for a player to switch its action in future time instants) that is useful to guarantee convergence~\cite{marden2009joint}. Thus, the results for convergence of standard fictitious play algorithm in an $n$ player game are weaker. As a side contribution potentially of independent interest, 
we define a notion of non-degeneracy in $n$ player finite games and show that the presence of an ordinal potential function assures the convergence of fictitious play to the respective Nash equilibrium in such games.

\paragraph{Literature Review}
\label{sec:lit_review}
Ever since its introduction~\cite{brown1951}, fictitious play (FP) has been a popular learning algorithm in game theory~\cite{milgrom1991a,fudenberg2008book}. The class of games for which the algorithm converges to Nash equilibrium has been gradually expanded (see, e.g.,~\cite{robinson1951a,miyazawa1951a,monderer1996,berger2005a}), although it is known that the convergence does not hold in general~\cite{shapley1964some}. We focus on the variant known as alternating fictitious play that was actually the algorithm originally proposed by Brown~\cite{brown1951,berger2007}. This algorithm converges to the pure Nash equilibrium for non-degenerate ordinal potential games with two players. In this work, we consider games with more than two players. For games with $n$ players, in \cite{marden2005}, a lack of absorption property for standard FP was illustrated even for potential games for $n$ players. As part of our proofs, we show that this absorption property can be revived by imposing an additional constraint. 

However, almost all the existing convergence proofs in the literature assume that all players update their strategies according to FP. The payoffs that a strategic player may be able to derive by deviating from the algorithm (even as the other players continue to play FP) is largely unexplored. In our formulation, the strategic player is assumed to possess the knowledge of the payoff matrices of all the players, while the other players know (or use) only their own payoffs for various strategy combinations in keeping with FP. Players having access to dissimilar information about the game is, of course, widely studied (e.g. as games of incomplete information, in the form of incredible threats in dynamic games, or through models of bounded rationality~\cite{anibal2016a}). However, less work has considered it in the context of learning in games. One relevant field that has studied one rational patient player playing against a boundedly rational opponent that employs myopic best response is Market Dominance and the Chain-Store Game~\cite{milgrom1982a} where it is known that a predation based strategy delivers a higher payoff in the long run by encouraging a {\em reputation} for the rational player. The authors in \cite{kreps1982} limit the response of the opponent to a finite history of plays and show that the predation effects change dramatically since the players ignore future consequences towards the end of the time horizon. 

If the strategic player announces her commitment to a Stackelberg strategy where she is the leader, she can obtain the corresponding payoff both in reputation based setups and if the opponent is implementing fictitious play. Along this theme, \cite{fudenberg1987a} identifies the conditions under which the results from sequential game play extend to its simultaneous counterpart. These results are improved in \cite{fudenberg1989a} that further accounts for the possibility that distinct strategies on the long-run player could be observationally equivalent. \cite{fudenberg1992} concludes that if public commitments are allowed, then the best that a long-run strategic player could do is to publicly commit to a pure Stackelberg strategy while the opponents take the role of Stackelberg followers. \cite{freund1995} extends the discussion to contract games while \cite{conitzer2006} shows that the Stackelberg strategy that the strategic player announces to the opponents may be mixed.

In contrast to this stream of work,  we do not allow communication among the players, so that the strategic player can no longer commit to or announce her strategy publicly. In this case, we show that the pure Stackelberg solution turns out to be a special case of the convergence based mixed strategy for the intelligent player, that delivers an even higher payoff. Specifically, we present a sequence of strategies for the intelligent player that not only converges to her convergence based mixed strategy  but also restricts the opponents to a desired strategy profile that delivers the intended higher payoff. 

\paragraph{Contributions}
We consider the interaction between $n+1$ players that play a matrix stage game repeatedly. The players are classified based on their information level where the first class consists of a single intelligent player \textit{(IP)} who is aware of the complete game. All the remaining players, referred to as opponents, belong to the second class and are limited to the knowledge of their own payoffs for different strategy vectors. When all players employ FP, under suitable conditions, the players converge to the Nash equilibrium. However, the IP need not adhere to Fictitious Play. We ask the question: \textit{Can the IP obtain a higher than Nash equilibrium payoff by deviating from fictitious play?} Further, if there exists such a strategy profile, \textit{how does the IP enforce it when the opponents are implementing fictitious play? } Compared with prior works, our key contributions are as follows.
\begin{enumerate}
	\item We identify strategies that can deliver an expected payoff greater than the Nash and the Stackelberg equilibrium payoff for the \textit{IP}. For the case when there are 2 players in the game, the strategies that we identify are optimal for the {\em IP}. For the general case of $n+1$ players, we provide a more tractable class of strategies that we term as convergence based mixed strategies that may be sub-optimal, yet can provide an expected payoff greater than the Nash and the Stackelberg payoff for the \textit{IP}.
	\item We provide a Linear Programming formulation that determines the strategy identified above without having to explore actions of all opponents at every time instant. 
	\item We determine a pure action trajectory for the \textit{IP} that reaches the desired mixed strategy probabilities while keeping the opponents in their FP determined strategies. 
\end{enumerate}

The remainder of this article is structured as follows. Section~\ref{sec:problemdef} presents the basic problem setup and summarizes the alternating Fictitious Play 
algorithm. Section \ref{sec:twoplayers} presents \textit{IP}'s optimal strategy trajectory to obtain the highest possible payoff against a single opponent. Section \ref{sec:problem_n} provides a computationally tractable sub-optimal solution for a game with more than two players. Section~\ref{sec:conclusions} concludes the paper. 

\section{Problem Description}
\label{sec:problemdef}


Consider a finite game $\mathbf{G}=(n+1, {Y_i}, {U_i})$ with $n+1$ players, where each player $\mathcal{P}_i \in \{\mathcal{P}_0, \mathcal{P}_1, \cdots, \mathcal{P}_n\}$  has an action set $Y_i$ and a utility function $U_i: Y\rightarrow \mathbb{R}$ where $Y \coloneqq Y_0 \times Y_1 \times \cdots \times Y_n$. Further, for a given action profile $y = (y_0, y_1, \cdots, y_n) \in Y$, let $y_{-i}\coloneqq (y_0, \cdots, y_{i-1}, y_{i+1}, \cdots, y_n)$ denote a profile of player actions other than player $\mathcal{P}_i$, 
With a slight abuse of notation, a profile $y$ of actions can be written as $(y_i, y_{-i})$ and the corresponding utilities $U_i(y)$  as $U_i(y_i, y_{-i})$. We assume that the game $\mathbf{G}$ is played at times $t=1, 2, 3, \cdots$. 
A mixed strategy $\tilde{\mathbf{z}}_i$ is a vector of probabilities for all actions in the set $Y_i$. Denote  by $\tilde{\mathbf{z}}_{-i}$ the profile of mixed strategies for all players other than $\mathcal{P}_i$  
and the expected utility for $\mathcal{P}_i$ playing a pure action $y_i$ and the rest of the players playing  $\tilde{\mathbf{z}}_{-i}$ by $U_i(y_i,\tilde{\mathbf{z}}_{-i})$. When a game is played repeatedly, the mixed strategy vector of player $\mathcal{P}_i$ can change with time and  we denote the vector at time $t$  by $\tilde{\mathbf{z}}_i(t)$.

The players are categorized based on their information structure within the repeated game. In the first category, players are aware of only their own payoffs and the actions of all the players as realized in all stage games till that time.
We refer to them as the opponents and denote them using $\mathcal{P}_j$. Without loss of generality, we assume that $\mathcal{P}_j\in \mathcal{P}'\coloneqq \{\mathcal{P}_1,\mathcal{P}_2,\cdots,\mathcal{P}_n\}$. We assume that the opponents adhere to Alternating Fictitious Play (see Definition \ref{def:fict_play_n} below), that offers convergence guarantees in spite of the limited information availability at the players. Player $\mathcal{P}_0$ falls in the second category by dint of her knowledge of the entire game $\mathbf{G}$ which includes the payoff matrix for all the players, as well as her knowledge of her own payoff and actions of all the players as realized in all stage games till that time. We consider $\mathcal{P}_0$ to be the Intelligent Player and refer to it as the \textit{IP}. The \textit{IP} may deviate from FP to obtain a higher payoff. 
Every stage game is played as follows. At every stage $t$, the \textit{IP} begins with its action, say $y_0 \in Y_0$. The best response of all the remaining players then follows in the order of their indices, as specified  in Definition~\ref{def:fict_play_n}. Note that the assumption on the order is without loss of generality. Once all players have played at stage $t$, the payoff for all the players is realized after identifying the actions played by all the players at time $t$. The game then moves on to stage $t+1$.
\begin{definition}
	\label{def:fict_play_n}
	For a game $\mathbf{G} = (n+1, {Y_i}, {U_i})$, an opponent $\mathcal{P}_j \in \mathcal{P}'$ is considered to be adhering to \textbf{Alternating Fictitious Play (referred simply as FP in this paper)} if at every stage game at time $t$, $\mathcal{P}_j$ plays its best response (\textit{BR}$_i$($\cdot$), given by (\ref{eq:best_response_n})), to the empirical distribution of actions of players $\{\mathcal{P}_0, \mathcal{P}_1, \cdots, \mathcal{P}_{j-1}\}$ until time $t$ and that of players $\{\mathcal{P}_{j+1}, \cdots, \mathcal{P}_n\}$ until time $(t-1)$:
	\begin{equation}
	\label{eq:best_response_n}
		\begin{aligned}
			&{BR}_i(\hat{\tilde{\mathbf{z}}}_{-i}(t)) \coloneqq \argmax_{y_i\in Y_i}~ U_i(y_i,\hat{\tilde{\mathbf{z}}}_{-i}(t)),
			\end{aligned}
			\end{equation}
where $\hat{\tilde{\mathbf{z}}}_{-i}(t)= \{\hat{\tilde{\mathbf{z}}}_0(t), \cdots, \hat{\tilde{\mathbf{z}}}_{i-1}(t), \hat{\tilde{\mathbf{z}}}_{i+1}(t-1), \cdots, \hat{\tilde{\mathbf{z}}}_n(t-1)\}$ is the estimate of the mixed strategy profile of all the other players as calculated using the empirical frequency of the actions played by players $\{\mathcal{P}_0, \mathcal{P}_1, \cdots, \mathcal{P}_{j-1}\}$ until time $t$ and players $\{\mathcal{P}_{j+1}, \cdots, \mathcal{P}_n\}$ until time $(t-1)$.
	\hfill $\blacksquare$
\end{definition}

Note that how the {\em IP} should choose its strategy has not been specified. If she chooses FP as well, for many classes of games $\mathbf{G}$, the players will converge to the Nash equilibrium strategies. However, the {\em IP} can potentially obtain a better payoff by deviating from FP. 
The problem we are interested in  is to identify the optimal strategy for the {\em IP} to obtain the best payoff when the opponents continue to play FP. It is not clear a priori, whether the optimal strategy for the {\em IP} will be a mixed or a pure strategy, and whether the game will converge to an equilibrium or not when the {\em IP} deviates from FP. Further, we are interested in identifying an action sequence for the {\em IP} that realizes the desired strategy profile for all the players when the opponents play FP. 

\section{Two Player Games}
\label{sec:twoplayers}

The analysis of convergence of FP is much simpler and more advanced in 2 player games. We begin with that special case as well and show that an optimal strategy for the {\em IP} can be calculated using a linear program without the need to assume any additional structure on the game $\mathbf{G}.$

Let the action set $Y_{0}$ of the {\em IP} be the set $Y_{0}=\{y_{0}^{1},y_{0}^{2},\cdots,y_{0}^{n}\}$ and the set $Y_{1}$ for the opponent be the set $Y_{1}=\{y_{1}^{1},y_{1}^{2},\cdots,y_{1}^{m}\}.$ We can then denote the payoffs of the {\em IP} (resp. the opponent) through an $n\times m$ matrix $A$ (resp. $B$) such that the $(i,j)$-th element $a_{ij}$ of $A$ (resp. $b_{ij}$ of $B$) denotes the payoff of the {\em IP} (resp. the opponent) when the {\em IP} chooses action $y_{0}^{i}$ and the opponent chooses action $y_{1}^{j}$. 
Consequently, the best response for the {\em IP} corresponding to the opponent playing a mixed strategy $\tilde{\mathbf{z}}_{1}$ is given by $\argmax_i (A\tilde{\mathbf{z}}_{1})_i$ and the best response for the opponent corresponding to the {\em IP} playing a mixed strategy $\tilde{\mathbf{z}}_{0}$ is given by $\argmax_j (B\tilde{\mathbf{z}}_{0})_j$ 

FP by both players in a 2 player game is known to converge to a Nash equilibrium in some specific games, for instance in $2 \times M$ games with generic payoffs. We emphasize that we do not impose such restrictions on the payoffs of the players. However,  we assume that the payoffs of the opponent are indexed such that a lower index indicates a higher payoff for the \textit{IP}, and to break ties in FP, the opponent employs a lower index when indifferent between two or more pure actions. Our first result notes that the optimal strategy for the {\em IP} restricts the opponent to play a pure strategy in the steady state. 
\begin{theorem}
	The optimal strategy for the \textit{IP} is such that the opponent  plays a single action in the steady state. 
	\label{theorem:stativ_vs_dynamic}
\end{theorem}
\begin{proof}
	At every time instant, the opponent employing FP has a unique pure strategy best response for any trajectory history of the \textit{IP}. Thus, any switching between opponent's pure strategies that deliver unequal payoffs to the \textit{IP} would reduce the expected payoff. To maximize its expected payoff, the {\em IP} thus restricts the opponent to the pure strategy that delivers her highest payoff. 
\end{proof}

We can then characterize the optimal strategy of the \textit{IP}. 

\begin{theorem}
	Let $\tilde{\mathbf{z}}_{0}^{(j)}$ be a (possibly mixed) strategy for the \textit{IP} such that the corresponding best response for the opponent is the pure strategy $y_{1}^{j}$ where $j \in \{1,2,\cdots,m\}$. The strategy profile $(\tilde{\mathbf{z}}_{0}^*,{y}_{1}^{j^*})$  that maximizes the expected payoff of the \textit{IP} is given as follows:
	\begin{equation}
		\begin{gathered}
			j^* = \argmax_j  \left(\max_{\tilde{\mathbf{z}}_{0}^{(j)}} (\tilde{\mathbf{z}}_{0}^{(j)}\cdot Ay_{1}^{j})\right) ~\forall j \in \{1,2,\cdots,m\} \\
			\tilde{\mathbf{z}}_{0}^* = \max_{\tilde{\mathbf{z}}_{0}^{(j^*)}} (\tilde{\mathbf{z}}_{0}^{(j^*)}\cdot Ay_{1}^{j^*}).
		\end{gathered}
		\label{eq:two_pl_mixed}
	\end{equation}
	\label{theorem:two_pl_mixed}
\end{theorem}
\begin{proof}
	By Theorem~\ref{theorem:stativ_vs_dynamic}, the {\em IP} can restrict her search over strategies that lead to best responses for the opponent that are pure strategies in the steady state. In other words, the only strategies of interest are the non-dominated pure strategies  of the opponent. Now for every such strategy $y_{1}^{j}$, there exist (possibly multiple) mixed strategies $\tilde{\mathbf{z}}_{0}^{(j)}$ such that the opponent's best response to the mixed strategy via Fictitious play is $y_{1}^{j}$ and the corresponding payoff for the \textit{IP} is given by $\tilde{\mathbf{z}}_{0}^{(j)}\cdot A{y}_{1}^{j}$. The highest payoff for the {\em IP} can then be computed by maximizing this payoff over the mixed strategy space, followed by identifying the pure action of the opponent with the highest such payoff, as given by (\ref{eq:two_pl_mixed}). 	
\end{proof}

\begin{remark}
	Theorem \ref{theorem:two_pl_mixed} can be restated as $m$ Linear Programming problems with constraints arising from restricting the response of the opponent to one of the pure actions while the maximization comes from the expected payoff of the \textit{IP}. 
	\label{rem:lp_two}
		\hfill $\blacksquare$
\end{remark}

\begin{example}
\label{eg:play_2_strat_2X3}
\begin{table}[bhtp]
	\centering
	\begin{tabular}{ccccc}
		&                          & \multicolumn{3}{c}{$2$ (Opponent)}                                                              \\
		&                          & $L$                            & $B$                            & $R$                            \\ \cline{3-5} 
		\multirow{2}{*}{$1$ (\textit{IP})} & \multicolumn{1}{c|}{$U$} & \multicolumn{1}{c|}{({\bf 6},{\bf 10})} & \multicolumn{1}{c|}{(10,7)} & \multicolumn{1}{c|}{({\bf 8},2)} \\ \cline{3-5} 
		& \multicolumn{1}{c|}{$D$} & \multicolumn{1}{c|}{(5,1)} & \multicolumn{1}{c|}{({\bf 15},8)} & \multicolumn{1}{c|}{(7,{\bf 9})} \\ \cline{3-5} \\
	\end{tabular}
	\caption{A finite two player game considered in Example~\ref{eg:play_2_strat_2X3}}
	\label{tab:play_2_strat_2X3}
\end{table}

A finite two player game with two pure actions for the \textit{IP} (the row player) and three pure actions for the opponent (the column player) is presented via Table \ref{tab:play_2_strat_2X3}, where the best responses for each player are marked in bold. When both players employ FP, the game converges to the pure Nash equilibrium $(U,L)$, where the \textit{IP} obtains an expected payoff of $6$. To turn the game into a Stackelberg game, the \textit{IP} can play $D$ repeatedly and shift the game to $(D,R)$, thus obtaining a payoff of $7$. Theorem \ref{theorem:two_pl_mixed} increases this payoff further by posing three Linear programming problems corresponding to the three pure actions of the opponent. It turns out that a mixed strategy of $(\frac{1}{6}, \frac{5}{6})^T$ restricts the opponent to the pure action $B$ and delivers an expected payoff of $14.17$, greater than prior solutions and is also the highest possible payoff. The conditions from Linear Programming problem require the \textit{IP} to maintain the probabilities of its pure action $U$ in the range $[\frac{1}{6}, \frac{7}{10}]$, in order to restrict the opponent to $B$. A strategy trajectory that achieves this is $(U, D, D, D, D, D)$. 

\end{example}

\section{Games with more than two players}
\label{sec:problem_n}


While the above analysis can be generalized to games with more than two players, the solution quickly becomes computationally complicated. Further, analysis of convergence of FP in games with more than two players is more limited than in those with two players even if no {\em IP} is present.  We now make some assumptions on the game structure and present a suboptimal but computationally more tractable solution.

\paragraph{Assumptions on the Game} Define a subgame $\mathbf{G}^{(y_0)}$ that restricts the \textit{IP} to one of its pure actions $y_0$ as follows (note that the subscript $i$ is reserved for all players in game $\mathbf{G}$ while $j$ is reserved for opponents in subgame $\mathbf{G}^{(y_0)}$ i.e., $\mathcal{P}_i \in \mathcal{P}$ and $\mathcal{P}_j \in \mathcal{P}'$):
		$\mathbf{G}^{(y_0)} = (n, Y_j, U^{(y_0)}_j)	~\text{where}~ U^{(y_0)}_j = U_j(y_0,y_j,y_{-j})
		\text{ such that}~ j \in \{1,\cdots, n\}, y_j \in Y_j, y_{-j} \in \times_{s\neq \{0,j\}} Y_s.$
%
We assume that the game $\mathbf{G}$ is a non-degenerate ordinal  potential game {\em with respect to the IP} as defined below.

\begin{definition}	
	\label{def:degeneracy_n}
	A game $\mathbf{G} = (n+1, {Y_i}, {U_i})$ is considered to be degenerate with respect to the \textit{IP}, if for some $y_0 \in Y_0$, there exists $y_{-0}, y'_{-0} \in\times Y_j$ such that $U_j(y_0,y_{-0}) = U_j(y_0,y'_{-0})$, if $y_{-0} \neq y'_{-0}$. Otherwise, the game is said to be \textbf{non-degenerate with respect to the \textit{IP}}. 
%
	Further, it is an \textbf{ordinal potential game with respect to the IP} if every subgame $\mathbf{G}^{(y_0)} = (n, Y_j, U^{(y_0)}_j)$, $y_0 \in Y_0$, is an ordinal potential game with a unique pure Nash equilibrium. 
	\hfill $\blacksquare$
\end{definition}
Note that we restrict our discussion to non-degenerate and ordinal potential games where the constraints are applicable only to the subgames  $\mathbf{G}^{(y_0)}$   and not to  $\mathbf{G}$ itself. Consequently, the discussion below relates to a larger class of games. 
We first show a convergence result for FP in this larger class of games.
\begin{theorem}
	In every subgame $\mathbf{G}^{(y_0)}$ of a finite non-degenerate ordinal potential game $\mathbf{G}$ with respect to the \textit{IP}, FP 
	by the opponents converges to a pure Nash equilibrium of the game in a finite number of time steps. 
	\label{theorem:AFP_Nash}
\end{theorem}

\begin{proof}
	In a subgame with finite strategy profiles, there cannot be an infinite sequence of improvement steps without resorting to cycles. Since such cycles are absent in ordinal potential subgames, any improvement path converges to a strategy profile. In a non-degenerate subgame, such a strategy profile 
	does not have any improvement steps leading out of it, making it the pure Nash equilibrium. As a result, all improvement paths converge to a pure Nash equilibrium in finite time. 
\end{proof}

\paragraph{Determination of the best convergence based mixed strategy for the \textit{IP}}  In order to determine a strategy that increases {\em IP}'s payoff, one possibility is to consider the payoffs in the subgames $\mathbf{G}^{(y_{0})}$ corresponding to all its pure actions and select the one that yields the best payoff. However, the {\em IP} can in fact do better by switching between this action and others in a manner that increases her payoff, without allowing the opponents to switch from the pure Nash equilibrium of the subgame. We may term such a strategy for the {\em IP} as a \textbf{convergence based mixed strategy}, which is a mixed strategy specific to a pure action $y_0$, such that the \textit{IP} switches between its pure actions while restricting the opponents to the pure Nash equilibrium of the subgame $\mathbf{G}^{(y_0)}$. By assumption, the opponents switch from one action to another only if the expected payoff for the former is strictly lower than the latter. 

Let $\mathbf{z}_0$ denote \textit{IP}'s convergence based mixed strategy such that the best response for all the opponents is $(y_0, y_j^*, y_{-j}^*)$ in the subgame $\mathbf{G}^{(y_0)}$. The expected payoff for any player $\mathcal{P}_i \in \mathcal{P}$ is given by (\ref{eq:payoff_IP_n}), where $\mathbf{z}_0^k$ is the probability corresponding to a pure action $k \in Y_0$. 
\begin{equation}
	U_i^{(y_0)}(\mathbf{z}_0,y_j^*, y_{-j}^*) = \sum_{k \in Y_0} \mathbf{z}_0^k ~U_i(k,y_j^*, y_{-j}^*).
	\label{eq:payoff_IP_n}	
\end{equation}

The following result 
determines the subgame $\mathbf{G}^{(y_0^*)}$ and its corresponding convergence based mixed strategy $\mathbf{z}_0^*$. 

\begin{theorem}
	Let $(y_0, y_j^*, y_{-j}^*)$ be the strategy profile corresponding to the pure Nash equilibrium for the subgame $\mathbf{G}^{(y_0)}$, where $y_0 \in Y_0$ and $j\in \{1,2,\cdots,n\}$. Let $U_0^{(y_0)}(\mathbf{z}_0,y_j^*,y_{-j}^*)$ be the expected payoff for \textit{IP} for its mixed strategy $\mathbf{z}_0$ in the subgame $\mathbf{G}^{(y_0)}$, as given by (\ref{eq:payoff_IP_n}). The strategy profile $(\mathbf{z}_0^*,y_j^*, y_{-j}^*)$  that maximizes the expected payoff of the \textit{IP} is:  
	\begin{equation}
		\begin{gathered}
			y_0^* = \argmax_{y_0}  (\max_{\mathbf{z}_0} ~U_0^{(y_0)}(\mathbf{z}_0,y_j^*,y_{-j}^*))\\ 
			\mathbf{z}_0^* = \argmax_{\mathbf{z}_0} ~U_0^{(y_0^*)}(\mathbf{z}_0,y_j^*,y_{-j}^*).
		\end{gathered}
		\label{eq:n_pl_mixed}
	\end{equation}
	\label{theorem:n_pl_mixed}
\end{theorem}
\begin{proof}
From Theorem \ref{theorem:AFP_Nash}, it is adequate to consider the Nash equilibria of the opponents in lieu of the remaining strategy profiles since FP 
	always converges to the pure Nash equilibria in the class of games we consider. For every subgame $\mathbf{G}^{(y_0)}$ and the corresponding pure Nash equilibrium $(y_0, y_j^*,y_{-j}^*)$, there exists a mixed strategy profile $\mathbf{z}_0$ for the \textit{IP} whose best response for all the opponents is still the Nash equilibrium profile $(y_0, y_j^*,y_{-j}^*)$. 
	Thus, the maximum expected payoff for the \textit{IP} can be computed by first maximizing such payoff $U_0^{(y_0)}(\mathbf{z}_0,y_j^*,y_{-j}^*)$, within a subgame $\mathbf{G}^{(y_0)}$ followed by identifying the subgame $\mathbf{G}^{(y_0^*)}$ with the highest maximum expected payoff, as given by (\ref{eq:n_pl_mixed}). The corresponding convergence based mixed strategy $\mathbf{z}_0^*$ can be obtained by using the subgame $\mathbf{G}^{(y_0^*)}$ and maximizing the expectation from (\ref{eq:payoff_IP_n}). 
\end{proof}

We now present a result that indicates a procedure to compute $\mathbf{z}_0^*$ via a linear program (LP). 

\begin{corollary}
	The computation of each $\mathbf{z}_0$ via Theorem \ref{theorem:n_pl_mixed} for a pure action $y_0$ of the \textit{IP} can be solved using an LP. $\mathbf{z}_0^*$ can be calculated by solving a number of such problems equal to the cardinality of the set $Y_0$ and then by choosing the subgame with the highest maximum expected payoff. 
	\label{cor:n_pl_mixed}
\end{corollary}
\begin{proof}
	The proof follows from the structure of (\ref{eq:n_pl_mixed}) where $\mathbf{z}_0$ is obtained by maximizing the cost function $U_0^{(y_0)}(\mathbf{z}_0,y_j^*,y_{-j}^*)$. Specifically, the Linear Programming problem for a pure action $k \in Y_0$ can be stated via (\ref{eq:LP}), where, for brevity, $\mathbf{z}_0$ is denoted by a vector $\mathbf{q} = (q_1, q_2, \cdots, q_{|Y_0|})^T$ where $|Y_0|$ is cardinality of the set $Y_0$ and the strategy profile $(k,y_j^*, y_{-j}^*)$ is the Nash equilibrium for opponents in the subgame $\mathbf{G}^{(k)}$. 
\begin{equation}
	\begin{gathered}
		\text{Maximize} \sum_{k \in Y_0} q_k ~U_0(k,y_j^*, y_{-j}^*) \\
		\text{s.t.} \sum_{k \in Y_0} q_k ~[U_j(k,y'_j,y_{-j}^*) - U_j(k,y_j^*,y_{-j}^*)] \leq 0 \\
		\forall ~ y'_j \in Y_j, y'_j \neq y_j^* ~\text{and}~ \forall ~j \in \{1,2,\cdots, n\} \\
		\sum_{k \in Y_0} q_k = 1, ~~q_k \in [0,1].
	\end{gathered}
	\label{eq:LP}
\end{equation}
Various constraints for the Linear Programming problem arise from restricting the opponents' best responses to the pure Nash equilibrium of the subgame $\mathbf{G}^{(y_0)}$. The additional maximization in (\ref{eq:n_pl_mixed}) then leads to the identification of the subgame $\mathbf{G}^{(y_0^*)}$ that delivers the maximum expected payoff. 
\end{proof}

\begin{example}
	A finite non-degenerate ordinal potential three player game with respect to the \textit{IP} ($\mathcal{P}_0$) is presented via Table \ref{tab:play_3_strat_3}. Each player has three strategies and each matrix represents the game with respect to one of the strategies of the \textit{IP}. The row player is $\mathcal{P}_1$ while the column player is $\mathcal{P}_2$. 
	\begin{table}[bhtp]
		\centering
		\begin{tabular}{ccccc}
			&                          & \multicolumn{3}{c}{$\mathcal{P}_2$}                                                              \\
			&                          & $L$                            & $N$                            & $R$                            \\ \cline{3-5} 
			\multirow{3}{*}{$\mathcal{P}_1$} & \multicolumn{1}{c|}{$U$} & \multicolumn{1}{c|}{(2,1,1) {\color{ForestGreen} \bf 1}} & \multicolumn{1}{c|}{(3,{\bf 6},{\bf 3}) {\color{ForestGreen} \bf 9}} & \multicolumn{1}{c|}{(6,8,2) {\color{ForestGreen} \bf 5}} \\ \cline{3-5} 
			& \multicolumn{1}{c|}{$M$} & \multicolumn{1}{c|}{(3,2,7) {\color{ForestGreen} \bf 2}} & \multicolumn{1}{c|}{(4,4,8) {\color{ForestGreen} \bf 3}} & \multicolumn{1}{c|}{(3,7,{\bf 9}) {\color{ForestGreen} \bf 4}} \\ \cline{3-5} 
			& \multicolumn{1}{c|}{$D$} & \multicolumn{1}{c|}{(3,{\bf 3},5) {\color{ForestGreen} \bf 7}} & \multicolumn{1}{c|}{(2,5,{\bf 6}) {\color{ForestGreen} \bf 8}} & \multicolumn{1}{c|}{(4,{\bf 9},4) {\color{ForestGreen} \bf 6}} \\ \cline{3-5} 
			&                          & \multicolumn{3}{c}{$\mathcal{P}_0 \Leftrightarrow \text{\textit{IP}}~(y_0 = A)$}                 \\
			&                          & $L$                            & $N$                            & $R$                            \\ \cline{3-5} 
			\multirow{3}{*}{$\mathcal{P}_1$} & \multicolumn{1}{c|}{$U$} & \multicolumn{1}{c|}{(2,4,2) {\color{ForestGreen} \bf 3}} & \multicolumn{1}{c|}{(4,3,3) {\color{ForestGreen} \bf 4}} & \multicolumn{1}{c|}{(7,6,{\bf 9}) {\color{ForestGreen} \bf 5}} \\ \cline{3-5} 
			& \multicolumn{1}{c|}{$M$} & \multicolumn{1}{c|}{(4,2,4) {\color{ForestGreen} \bf 2}} & \multicolumn{1}{c|}{(3,1,1) {\color{ForestGreen} \bf 1}} & \multicolumn{1}{c|}{(3,8,{\bf 5}) {\color{ForestGreen} \bf 6}} \\ \cline{3-5} 
			& \multicolumn{1}{c|}{$D$} & \multicolumn{1}{c|}{(3,{\bf 7},{\bf 8}) {\color{ForestGreen} \bf 9}} & \multicolumn{1}{c|}{(2,{\bf 5},7) {\color{ForestGreen} \bf 8}} & \multicolumn{1}{c|}{(4,9,6) {\color{ForestGreen} \bf 7}} \\ \cline{3-5} 
			&                          & \multicolumn{3}{c}{$\mathcal{P}_0 \Leftrightarrow \text{\textit{IP}}~(y_0 = B)$}                 \\
			&                          & $L$                            & $N$                            & $R$                            \\ \cline{3-5} 
			\multirow{3}{*}{$\mathcal{P}_1$} & \multicolumn{1}{c|}{$U$} & \multicolumn{1}{c|}{(2,1,1) {\color{ForestGreen} \bf 1}} & \multicolumn{1}{c|}{(5,{\bf 9},2) {\color{ForestGreen} \bf 8}} & \multicolumn{1}{c|}{(5,{\bf 6},{\bf 3}) {\color{ForestGreen} \bf 9}} \\ \cline{3-5} 
			& \multicolumn{1}{c|}{$M$} & \multicolumn{1}{c|}{(4,2,7) {\color{ForestGreen} \bf 2}} & \multicolumn{1}{c|}{(3,8,{\bf 9}) {\color{ForestGreen} \bf 7}} & \multicolumn{1}{c|}{(4,5,8) {\color{ForestGreen} \bf 6}} \\ \cline{3-5} 
			& \multicolumn{1}{c|}{$D$} & \multicolumn{1}{c|}{(3,{\bf 3},4) {\color{ForestGreen} \bf 3}} & \multicolumn{1}{c|}{(2,7,5) {\color{ForestGreen} \bf 4}} & \multicolumn{1}{c|}{(3,4,{\bf 6}) {\color{ForestGreen} \bf 5}} \\ \cline{3-5} 
			&                          & \multicolumn{3}{c}{$\mathcal{P}_0 \Leftrightarrow \text{\textit{IP}}~(y_0 = C)$}                
		\end{tabular}
		\caption{A finite three player game $\mathbf{G}$ with three subgames}
		\label{tab:play_3_strat_3}
	\end{table}
	\label{eg:play_3_strat_3_basic}

In this example, there are three subgames $\mathbf{G}^{(A)}$, $\mathbf{G}^{(B)}$ and $\mathbf{G}^{(C)}$. Best responses of each opponent against pure actions of the remaining players, as given by (\ref{eq:best_response_n}) are marked in bold. Table \ref{tab:play_3_strat_3} also presents the potentials (marked in green) of various strategy profiles within their respective subgames. Naturally, the profile with the highest potential is the unique pure Nash equilibrium within that subgame. Each cell depicts the payoffs for the three players for that strategy profile. When all the players (including the \textit{IP}) employ FP, the game converges to $(B,D,L)$ and the \textit{IP} obtains an expected payoff of $3$  $(=U_0^*(y))$. To increase its payoff, 
the {\em IP} utilizes Corollary \ref{cor:n_pl_mixed} to formulate three LPs followed by identifying the pure action $y_0^*$ whose convergence based mixed strategy $\mathbf{z}_0^*$ delivers the required payoff. 
This $y_0^*$ turns out to be $C$ and $\mathbf{z}_0^*$ turns out to be $(\frac{1}{9}, \frac{3}{9}, \frac{5}{9})^T$. The corresponding LP is given via (\ref{eq:eg_2_LP}), where vector $\mathbf{q} = (q_1, q_2, q_3)^T$ indicates the probabilities of pure actions $(A, B, C)$. The constraints in (\ref{eq:eg_2_LP_const_1}$-$\ref{eq:eg_2_LP_const_4}) restrict the opponents to $(U,R)$ while (\ref{eq:eg_2_LP_max}) maximizes the expected payoff for the \textit{IP}. The \textit{IP} obtains a payoff of $5.78(=\frac{52}{9})$ with this particular strategy, which is greater than $U_0^*(y)=3$. 
\begin{subequations}
	\label{eq:eg_2_LP}	
	\begin{align}
		\max_{\mathbf{q}} ~~ 6 q_1 + 7 q_2 + 5 q_3	\label{eq:eg_2_LP_max}\\
		\text{s.t.} 
		~~ - q_1 + 2 q_2 - q_3   \leq  0 \label{eq:eg_2_LP_const_1}\\
		~~   q_1 + 3 q_2 - 2 q_3 \leq  0 \label{eq:eg_2_LP_const_2} \\
		~~ - q_1 - 7 q_2 - 2 q_3 \leq  0 \label{eq:eg_2_LP_const_3} \\
		~~   q_1 - 6 q_2 - q_3   \leq  0 \label{eq:eg_2_LP_const_4} \\
		q_1 +  q_2 + q_3= 1\\ 
		q_1, q_2, q_3 \in [0,1].
	\end{align}
\end{subequations}

\end{example}
Convergence based mixed strategies are optimal in a class of strategies as stated below.
\begin{corollary}
	\label{cor:pure_action_highest_payoff}
	When the \textit{IP} is restricted to playing a single pure action repeatedly, (\ref{eq:n_pl_mixed}) in Theorem \ref{theorem:n_pl_mixed} delivers the highest payoff for the \textit{IP}. 
\end{corollary}

\begin{proof}
	Since every opponent adheres to FP, the game converges to the pure Nash Equilibrium of $\mathbf{G}^{(y_0)}$ in finite steps, when \textit{IP} plays $y_0$ repeatedly. Thus, the opponents have no incentive to deviate from $(y_0, y_j^*, y_{-j}^*)$ later. 
	In a finite game, restricting convergence based mixed strategy to probabilities of pure action such that $\mathbf{z}_0^k = 1$ when $k=y_0$ and $\mathbf{z}_0^k = 0$ otherwise and maximizing the payoffs using (\ref{eq:n_pl_mixed}) delivers the subgame with the highest Nash equilibrium payoff.
\end{proof}

For instance, in Example \ref{eg:play_3_strat_3_basic}, three pure actions of the \textit{IP} lead to three different pure Nash equilibria $(A, U, N)$, $(B, D, L)$ and $(C, U, R)$ when the \textit{IP} is restricted to playing a pure action repeatedly. Since the expected payoffs are equivalent to the Nash equilibrium payoffs $3$, $3$ and $5$ respectively, it is possible to identify $C$ as the pure action for the \textit{IP} that delivers the highest payoff under such restriction, as given by Corollary \ref{cor:pure_action_highest_payoff}. 

\paragraph{Computation of the strategy trajectory for the \textit{IP}}

There exist multiple strategy trajectories that obey the probabilities in $\mathbf{z}_0^*$ and yet are incapable of delivering the expected payoff promised by Theorem \ref{theorem:n_pl_mixed}. Further, $\mathbf{X}$ is an infinite sequence that lacks ease in implementation, as it can only be computed with a careful examination of the cost function and individual constraints, at every time instant. We present below an algorithm to compute the strategy trajectory.

To this end, we divide the infinite trajectory into two parts, the first being a static finite sequence played only once followed by another finite sequence that is played repeatedly and indefinitely. We denote the two sequences as $\mathbf{X}' = (y_0(1), y_0(2), \cdots, y_0(\tau'))$ and $\mathbf{X}^* = (y_0(1), y_0(2), \cdots, y_0(\tau^*))$ where $\tau',\tau^* \in \mathbb{N}$ and $y_0 \in Y_0^*$ such that $\mathbf{X}'$ is played once followed by repeated play of $\mathbf{X}^*$. It is worth noting that the \textit{IP}'s expected payoff converges to the payoff obtained during the sequence $\mathbf{X}^*$. We begin with a result that identifies the entries in $\mathbf{X}'$. 

\begin{lemma}	
	For the subgame $\mathbf{G}^{(y^*_0)}$ computed via Theorem \ref{theorem:n_pl_mixed}, the sequence $\mathbf{X}$ begins with the repetition of the pure action $y_0^*$. As a result, $y_0^*$ is contained in the tuple $\mathbf{X}'$. 
	\label{lemma:support_rule}
\end{lemma}
\begin{proof}
	In order to maximize \textit{IP}'s expected payoff as per Theorem \ref{theorem:n_pl_mixed} and Corollary \ref{cor:n_pl_mixed}, it is desirable to restrict the opponents to the pure Nash equilibrium $(y_0^*, y_j^*, y_{-j}^*)$. Theorem \ref{theorem:AFP_Nash} indicates that in FP, 
	the strategies of opponents converge to this equilibrium in finite time steps, within the subgame $\mathbf{G}^{(y_0^*)}$. As a result, \textit{IP}'s pure action is restricted to $y_0^*$ until the opponents converge to $(y_0^*, y_j^*, y_{-j}^*)$, indicating that $y_0^* \in \mathbf{X}'$. 
\end{proof}

It has been established via the proof to Lemma \ref{lemma:support_rule} that repeated play of $y_0^*$ alone is adequate in converging all the opponents to the Nash equilibrium of $\mathbf{G}^{(y_0^*)}$. Consequently, the sequence $\mathbf{X}'$ is merely a repetition of $y_0^*$ for $\tau'$ time instants. In the first time instant, it is assumed that the opponents play a random pure action since FP 
needs at least one iteration to compute a best response. The value for $\tau'$ takes into account the number of maximum time instants required by the opponents to converge at $(y_0^*, y_j^*, y_{-j}^*)$. 
Further, FP 
of the opponents allows us to examine the expected payoffs of a single opponent while everyone else continues to adhere to their respective equilibria. 

We note that the support of $\mathbf{z}^*_0$ is the set of pure actions of the \textit{IP} with non-zero probabilities and denote it by $Y_0^*$ where $Y_0^* \subseteq Y_0$. The following result presents a connection between the coefficients (given by the matrix $\mathbf{A}$) of various pure actions in the constraint set of the form $\mathbf{A}\cdot \mathbf{q} \leq 0$. 

\begin{lemma}
	In the subgame $\mathbf{G}^{(y_0^*)}$, increasing the probabilities associated with pure actions other than $y_0^*$ increases the incentive for the opponents to deviate from $(y_0^*, y_j^*, y_{-j}^*)$. On the other hand, the probability for $y_0^*$ itself is inversely proportional to their incentive to deviate. 
	\label{lemma:opp_incentives}
\end{lemma}
\begin{proof}
	The proof follows from the structure of the game $\mathbf{G}$ where a strategy profile $(y'_0, y_j^*, y_{-j}^*), y'_0\in Y_0^*, y'_0\neq y^*_0$ is not necessarily a Nash equilibrium, making other strategy profiles in the subgame $\mathbf{G}^{(y'_0)}$ more attractive for the opponents. Mathematically, it follows from (\ref{eq:LP}) that an element $\mathbf{A}_{jk}$ is given by $(U_j(k,y'_j,y_{-j}^*) - U_j(k,y_j^*,y_{-j}^*))$. This expression is always negative for pure action $k=y_0^*$ while it can have positive entries for all other actions. Since the constraints are of the form $\mathbf{A}\cdot \mathbf{q} \leq 0$, all positive entries increase the incentive for opponents to deviate from $(y_0^*, y_j^*, y_{-j}^*)$ while the negative entries work in the opposite manner, proving Lemma \ref{lemma:opp_incentives}. 
\end{proof}

Once the players converge to $(y_0^*, y_j^*, y_{-j}^*)$, Theorem \ref{theorem:n_pl_mixed} and Corollary \ref{cor:n_pl_mixed} advocate the existence of a sequence $\mathbf{X}^*$ with the probability distribution given by $\mathbf{z}^*_0$ that restricts the opponents to $(y_0^*,y_j^*, y_{-j}^*)$, even when the \textit{IP} switches to actions other than $y_0^*$. It has already been shown, via Lemma \ref{lemma:opp_incentives}, that pure actions other than $y_0^*$ increase the incentive for opponents to deviate from $(y_0^*, y_j^*, y_{-j}^*)$. However, for any given opponent $\mathcal{P}_j$, the switch to a different action occurs only when the expected payoff from $U_j(k, y_j^*, y_{-j}^*)$ is strictly lesser than that from another pure action of $\mathcal{P}_j$ i.e., $U_j(k, y'_j, y_{-j}^*)$. 

Prior to determining $\mathbf{X}$, it is worth noting that any arbitrary sequence of actions in $\mathbf{X}^*$ with probabilities given by $\mathbf{z}_0^*$ need not deliver the desired payoff for the \textit{IP}. Consider, for instance, a sequence $\mathbf{X}^* = (A, B, B, B, C, C, C, C, C)$ in Example \ref{eg:play_3_strat_3_basic} with probabilities $(\frac{1}{9}, \frac{3}{9}, \frac{5}{9})^T$, derived via Theorem \ref{theorem:n_pl_mixed} and Corollary \ref{cor:n_pl_mixed}. Further, let $\mathbf{X}' = (C, C, C)$ to ensure convergence to the Nash equilibrium $(C, U, R)$. When $\mathbf{X}'$ is played first, followed by a repeated play of $\mathbf{X}^*$, \textit{IP} gets an expected payoff of $3.67$ as $\tau \rightarrow \infty$, as opposed to the desired $5.78$. We address this by commenting on the size of the sequences $\mathbf{X}^*$ and $\mathbf{X}'$, followed by a result that determines a candidate sequence $\mathbf{X}^*$. 

\begin{remark}
	\label{rem:tau*}
	The size of the tuple $\mathbf{X}^*$, given by $\tau^*$, is determined as the smallest integer that permits integer values to all pure actions $y_0 \in Y_0^*$ as they achieve their respective probabilities in $\mathbf{z}_0^*$. 	
 Further, 
	let $\tau_0$ be the maximum number of time instants required to converge opponents to the Nash equilibrium via FP.  
	The size of the sequence $\mathbf{X}'$, denoted by $\tau'$ is given by
	\begin{equation}
		\tau' = \max \{\tau_0, \tau^*\}.
		\label{eq:X'}
	\end{equation}
\end{remark}

For instance, when $\mathbf{z}_0^*$ is given by $(\frac{2}{3}, 0, \frac{1}{6}, \frac{1}{6})$, the smallest value for $\tau^*$ is $6$ which indicates the frequency of pure actions via $\tau^* \mathbf{z}_0^*$ is $(4, 0, 1, 1)$. Further, if the opponents need $3$ time instants to converge to the Nash equilibrium, (\ref{eq:X'}) provides a value of $6~(=\max \{3,6\})$ for $\tau'$. As a result, the sequence $\mathbf{X}'$ would be a repetition of the pure action $y_0^*$ for $6$ time instants. While Remark \ref{rem:tau*} 
provides a way to compute $\mathbf{X}'$ and $\tau^*$, they are not adequate in determining a sequence $\mathbf{X}^*$. 
We provide a result that achieves the desired strategy below.

\begin{theorem} 	
	\label{theorem:sequence}
	Let the cardinality of $Y_0^*$ be denoted by $m$ and the probability vector $\mathbf{z}_0^*$ be denoted by $\mathbf{q}^*$ where individual probability of a pure action $k_s (k_s\in Y_0^*)$ is given by $q^*_s$. (\ref{eq:X*}) and (\ref{eq:k1_restrict}) together illustrate a candidate sequence for $\mathbf{X}^*$ that achieves the twin purpose of (i) restricting the opponents to the Nash equilibrium $(y_0^*, y_j^*, y_{-j}^*)$ and (ii) generating expected payoff as given by Theorem \ref{theorem:n_pl_mixed} when $\tau \rightarrow \infty$. 
	\begin{equation}
		\begin{gathered}
			\mathbf{X}^* = (k_1(1), k_1(2), \cdots, k_1(\tau^*q^*_1),\\
			 k_2(\tau^*q^*_1+1), \cdots, k_2(\tau^*(q^*_1+q^*_2)), \cdots, \\
			 k_m(\tau^*\sum_{s=1}^{m-1}q^*_s+1), \cdots, k_m(\tau^*\sum_{s=1}^{m}q^*_s))
		\end{gathered}
		\label{eq:X*}
	\end{equation}
	\begin{equation}
		k_1 = y_0^* ~ \text{if}~ q^*_{y_0^*} \neq 0.
		\label{eq:k1_restrict}
	\end{equation}
\end{theorem}
\begin{proof}
	In order to prove (i), it is adequate to show that the proposed sequence $\mathbf{X}^*$ adheres to the constraints in (\ref{eq:LP}) at every time instant. It follows from Lemma \ref{lemma:opp_incentives} that the opponents' incentive to deviate from $(y_0^*, y_j^*, y_{-j}^*)$ decreases with repeated play of $y_0^*$ while it increases with the play of any other pure action in $Y_0^*$. As a result, the constraints in (\ref{eq:LP}) reduce to (\ref{eq:sequence_constraint}) at any given time instant, where $q_s(\tau)$ indicates the probability of a pure action $k_s$ until time $\tau$. 
	\begin{subequations}
		\begin{align}
			q_s(\tau) \geq q^*_s &~\text{where}~ k_s = y_0^* \label{eq:seq_const_y0*}\\
			q_s(\tau) \leq q^*_s &~\forall~ k_s \in Y_0^*\setminus \{y_0^*\}. \label{eq:seq_const_not_y0*}
		\end{align}
		\label{eq:sequence_constraint}
	\end{subequations}
	When $q^*_{y_0^*}\neq 0$, the first strategy in $\mathbf{X}^*$ is $y_0^*$ and is repeated for $\tau^*q^*_{y_0^*}$ time instants. Since $\mathbf{X}^*$ is repeated indefinitely, the probability $q_{y_0^*}$ of the pure action $y_0^*$ after $t\in\{1, 2, \cdots, \tau^*q^*_{y_0^*}\}$ time instants in $p^{th}$ repetition is given by (\ref{eq:seq_const_y0*_proof}). It can  be shown that $q_{y_0^*}-q^*_{y_0^*}\geq 0$ thus proving (\ref{eq:seq_const_y0*}). 
	\begin{equation}
		q_{y_0^*}=\frac{\tau'+p \tau^* q^*_{y_0^*} + t}{\tau' + p \tau^* + t}.
		\label{eq:seq_const_y0*_proof}
	\end{equation}
	
	For every other action $k_s\in Y_0^*\setminus y_0^*$, the probability $q_s$ after $(\tau^*q_{y_0^*}^*+t)$ time instants in $p^{th}$ repetition is given by (\ref{eq:seq_const_not_y0*_proof}), where $t\in\{1, 2, \cdots, \tau^*q_s^*\}$. Here, it is assumed that $k_s$ immediately follows $y_0^*$ in $\mathbf{X}^*$. 
	\begin{equation}
		q_s=\frac{p \tau^* q^*_s + t}{\tau' + p \tau^* + \tau^*q_{y_0^*}^*+t}
		\label{eq:seq_const_not_y0*_proof}
	\end{equation}	
	The relation between $\tau^*$ and $t$ is stated via (\ref{eq:t_tau*_rel}) while, for an arbitrary non-negative constant $\epsilon$, (\ref{eq:tau'tau*_rel}) illustrates a relation between $\tau'$ and $\tau^*$ that follows from (\ref{eq:X'}). 
	\begin{subequations}
		\begin{align}
			t \leq \tau^*q^*_s \leq \tau^* \label{eq:t_tau*_rel} \\
			\tau' = \tau^* + \epsilon. \label{eq:tau'tau*_rel}
		\end{align}
	\label{eq:relations_tau_t}
	\end{subequations}
	With the help of (\ref{eq:seq_const_not_y0*_proof}) and (\ref{eq:tau'tau*_rel}), $(q_s-q^*_s)$ can be stated via (\ref{eq:seq_const_not_y0*_proof_2}). The negativity of the first term in the numerator of (\ref{eq:seq_const_not_y0*_proof_2}) follows from (\ref{eq:t_tau*_rel}) while the second term is always negative, thus proving (\ref{eq:seq_const_not_y0*}) when $q^*_{y_0^*}\leq 0$.
	\begin{equation}
		q_s-q^*_s = \frac{(t-\tau^*q_s^*)-q_s^*(\epsilon+p \tau^* + \tau^*q_{y_0^*}^*+t)}{\tau' + p \tau^* + \tau^*q_{y_0^*}^*+t}.
		\label{eq:seq_const_not_y0*_proof_2}
	\end{equation}
	When additional pure actions are played between $y_0^*$ and $k_s$, the value for the denominator in (\ref{eq:seq_const_not_y0*_proof}) increases while the numerator remains unchanged. This makes the expression in (\ref{eq:seq_const_not_y0*_proof_2}) more negative and thus does not affect the conclusion. Finally, when $q^*_{y_0^*}=0$, the sequence $\mathbf{X}'$ still contains $y_0^*$ but the sequence $\mathbf{X}^*$ does not use $y_0^*$. As a result, (\ref{eq:seq_const_y0*_proof}) is inconsequential while (\ref{eq:seq_const_not_y0*_proof_2}) is still negative. This proves (\ref{eq:seq_const_not_y0*}) and thus (i) for the sequence given by (\ref{eq:X*}) and (\ref{eq:k1_restrict}). 
	
	Since the opponents are restricted to $(y_0^*, y_j^*, y_{-j}^*)$, the \textit{IP} gets its payoff promised by Theorem \ref{theorem:n_pl_mixed} at the end of the sequence $\mathbf{X}^*$ in every iteration. This is also evident from the fact that the probability vector for all pure actions within the sequence $\mathbf{X}^*$ is equal to the solution of the Linear Programming problem $(\mathbf{q} = \mathbf{q}^*)$ at the end of the sequence, indicating the maximum payoff. As $\tau\rightarrow \infty$, the payoff from $\mathbf{X}'$ becomes negligible since it is played only once and the expected payoff converges to the solution of Theorem \ref{theorem:n_pl_mixed} and Corollary \ref{cor:n_pl_mixed}, proving (ii) here. 
\end{proof}

The construction of an optimal trajectory is summarized in Algorithm~\ref{algo:compute_X}.  The length of the sequences $\mathbf{X}^*$ and $\mathbf{X}'$ is given by Remark \ref{rem:tau*}. 

\begin{algorithm}[bhtp]	
	\caption{Determination of $\mathbf{X}$ for the \textit{IP}}
	\label{algo:compute_X}
	\DontPrintSemicolon 
	\KwIn{Game $\mathbf{G} = (n+1, {Y_i}, {U_i})$ and pure Nash equilibria of subgames $\mathbf{G}^{(y_0)} ~\forall~ y_0 \in Y_0$}
	\KwOut{$\mathbf{z}_0^*$, $y_0^*$, $\mathbf{X}'$ and $\mathbf{X}^*$}
	\textbf{Initialize} $r \gets |Y_0|$, $c \gets 0$, list $\mathbf{V}[r]$, matrix $\mathbf{W}[r,r]$  \;
	\For{$s \gets 1$ \textbf{to} $r$}
	{
		Solve Linear Programming problem for $y_0^s$ via (\ref{eq:LP})\;
		Use the solution to obtain $\mathbf{W}[s,:] \gets \mathbf{q}$ \;
		$\mathbf{V}[s] \gets \sum_{k \in Y_0} q_k ~U_0(k,y_j^*, y_{-j}^*)$ \;
	}
	$s' \gets \argmax \mathbf{V}$\;
	\Return $y_0^*$ as $y_0^{s'}$ and $\mathbf{z}_0^*$ as $\mathbf{W}[s',:]$\;	
	$\tau^* \gets $ smallest integer that can enforce $\mathbf{z}_0^*$\;
	$\tau_0 \gets $ maximum number of time instants to converge to $(y_0^*, y_j^*, y_{-j}^*)$ in the subgame $\mathbf{G}^{(y_0^*)}$\;
	$\tau' = \max \{\tau_0, \tau^*\}$\;
	\Return $\mathbf{X}'$ as $\{y_0^*(1), y_0^*(2), \cdots, y_0^*(\tau')\}$\;
	\If{$q^*_{y_0^*} \neq 0$}
	{
		\For{$l \gets 1$ \textbf{to} $\tau^*q^*_{y_0^*}$}
		{			
			$c \gets c+1$; $\mathbf{X}^*(c) \gets y_0^*$\;
		}
	}
	\For{$s \gets y_0^1$ \textbf{to} $y_0^r$ \textbf{except} $y_0^*$}
	{
		\For{$l \gets 1$ \textbf{to} $\tau^*q^*_s$}
		{			
			$c \gets c+1$; $\mathbf{X}^*(c) \gets s$\;
		}
	}
	\Return $\mathbf{X}^*$\;
\end{algorithm}

For instance, in Example \ref{eg:play_3_strat_3_basic},  where $y_0^*$ turns out to be $C$ and the mixed strategy $\mathbf{z}_0^*$ is given by $(\frac{1}{9},\frac{3}{9},\frac{5}{9})^T$, the smallest integer that can achieve $\mathbf{X}^*$ is $9$. Since convergence to Nash equilibrium from any strategy profile can be obtained in $\tau_0 = 4$ time instants, $\tau'$ is set to be $9(=\max\{4,9\})$, per Remark \ref{rem:tau*}.
It follows from Theorem \ref{theorem:n_pl_mixed} and $q^*_{y_0^*}\neq0$ that the sequence $\mathbf{X}^*$ begins with $y_0^*=C$ and repeats until its probability is reached within $\mathbf{X}^*$. This is followed by the pure actions $A$ and $B$ until their respective probabilities are achieved. 
It is worth noting that this sequence generates negative values for the expressions in (\ref{eq:eg_2_LP}) proving that the constraints are valid at every time instant within the first iteration of $\mathbf{X}^*$. Theorem \ref{theorem:sequence} further proves  that the constraints and probabilities remain valid perpetually. In summary, the sequence $\mathbf{X}$ obtained from $\mathbf{X}' = (C, C, C, C, C, C, C, C, C)$ and $\mathbf{X}^* = (C, C, C, C, C, B, B, B, A)$ is indeed an optimal trajectory and generates an expected payoff of $5.78$ as $\tau \rightarrow \infty$. 

\begin{example}
	A finite non-degenerate ordinal potential three player game with respect to the \textit{IP} ($\mathcal{P}_0$) is presented via Table \ref{tab:play_3_strat_3_advanced}. Each player has three pure actions and each matrix represents the game with respect to one of the pure actions of the \textit{IP}. The row player is $\mathcal{P}_1$ while the column player is $\mathcal{P}_2$. 
	\begin{table}[bhtp]
		\centering
		\begin{tabular}{ccccc}
			&                          & \multicolumn{3}{c}{$\mathcal{P}_2$}                                                              \\
			&                          & $L$                            & $N$                            & $R$                            \\ \cline{3-5} 
			\multirow{3}{*}{$\mathcal{P}_1$} & \multicolumn{1}{c|}{$U$} & \multicolumn{1}{c|}{(6,1,1) {\color{ForestGreen} \bf 1}} & \multicolumn{1}{c|}{(1,5,3) {\color{ForestGreen} \bf 5}} & \multicolumn{1}{c|}{(5,{\bf 9},{\bf 4}) {\color{ForestGreen} \bf 9}} \\ \cline{3-5} 
			& \multicolumn{1}{c|}{$M$} & \multicolumn{1}{c|}{(1,{\bf 3},5) {\color{ForestGreen} \bf 3}} & \multicolumn{1}{c|}{(1,4,6) {\color{ForestGreen} \bf 4}} & \multicolumn{1}{c|}{(1,8,{\bf 9}) {\color{ForestGreen} \bf 8}} \\ \cline{3-5} 
			& \multicolumn{1}{c|}{$D$} & \multicolumn{1}{c|}{(1,2,2) {\color{ForestGreen} \bf 2}} & \multicolumn{1}{c|}{(4,{\bf 6},7) {\color{ForestGreen} \bf 6}} & \multicolumn{1}{c|}{(1,7,{\bf 8}) {\color{ForestGreen} \bf 7}} \\ \cline{3-5} 
			&                          & \multicolumn{3}{c}{$\mathcal{P}_0 \Leftrightarrow \text{\textit{IP}}~(y_0 = A)$}                 \\
			&                          & $L$                            & $N$                            & $R$                            \\ \cline{3-5} 
			\multirow{3}{*}{$\mathcal{P}_1$} & \multicolumn{1}{c|}{$U$} & \multicolumn{1}{c|}{(5,{\bf 6},{\bf 7}) {\color{ForestGreen} \bf 9}} & \multicolumn{1}{c|}{(1,8,3) {\color{ForestGreen} \bf 5}} & \multicolumn{1}{c|}{(6,{\bf 7},5) {\color{ForestGreen} \bf 8}} \\ \cline{3-5} 
			& \multicolumn{1}{c|}{$M$} & \multicolumn{1}{c|}{(1,2,4) {\color{ForestGreen} \bf 1}} & \multicolumn{1}{c|}{(1,{\bf 9},6) {\color{ForestGreen} \bf 6}} & \multicolumn{1}{c|}{(1,5,{\bf 9}) {\color{ForestGreen} \bf 7}} \\ \cline{3-5} 
			& \multicolumn{1}{c|}{$D$} & \multicolumn{1}{c|}{(1,4,1) {\color{ForestGreen} \bf 2}} & \multicolumn{1}{c|}{(3,3,{\bf 8}) {\color{ForestGreen} \bf 3}} & \multicolumn{1}{c|}{(1,1,2) {\color{ForestGreen} \bf 4}} \\ \cline{3-5} 
			&                          & \multicolumn{3}{c}{$\mathcal{P}_0 \Leftrightarrow \text{\textit{IP}}~(y_0 = B)$}                 \\
			&                          & $L$                            & $N$                            & $R$                            \\ \cline{3-5} 
			\multirow{3}{*}{$\mathcal{P}_1$} & \multicolumn{1}{c|}{$U$} & \multicolumn{1}{c|}{(8,6,8) {\color{ForestGreen} \bf 5}} & \multicolumn{1}{c|}{(1,2,1) {\color{ForestGreen} \bf 1}} & \multicolumn{1}{c|}{(9,8,{\bf 9}) {\color{ForestGreen} \bf 6}} \\ \cline{3-5} 
			& \multicolumn{1}{c|}{$M$} & \multicolumn{1}{c|}{(1,4,{\bf 7}) {\color{ForestGreen} \bf 4}} & \multicolumn{1}{c|}{(1,3,5) {\color{ForestGreen} \bf 2}} & \multicolumn{1}{c|}{(1,1,6) {\color{ForestGreen} \bf 3}} \\ \cline{3-5} 
			& \multicolumn{1}{c|}{$D$} & \multicolumn{1}{c|}{(1,{\bf 7},3) {\color{ForestGreen} \bf 8}} & \multicolumn{1}{c|}{(2,{\bf 5},{\bf 4}) {\color{ForestGreen} \bf 9}} & \multicolumn{1}{c|}{(1,{\bf 9},2) {\color{ForestGreen} \bf 7}} \\ \cline{3-5} 
			&                          & \multicolumn{3}{c}{$\mathcal{P}_0 \Leftrightarrow \text{\textit{IP}}~(y_0 = C)$}                
		\end{tabular}
		\caption{A finite three player game $\mathbf{G}$ with three subgames}
		\label{tab:play_3_strat_3_advanced}
	\end{table}
	\label{eg:play_3_strat_3_advanced}

In order to identify $U_0^*(y)$ for the \textit{IP} in Example \ref{eg:play_3_strat_3_advanced}, we assume that all players employ FP in the sequence of their indices. Unlike Example \ref{eg:play_3_strat_3_basic}, the strategies do not converge to either a single or a mixed strategy profile. This is not surprising since the game $\mathbf{G}$ itself is not a non-degenerate ordinal potential game. 
However, the game $\mathbf{G}$ is still within the class of games we consider in this paper since the subgames adhere to Definition \ref{def:degeneracy_n}. While computation of $U_0^*(y)$ is not a part of Algorithm \ref{algo:compute_X}, the difference between the expected payoff from $\mathbf{X}$ and the value of $U_0^*(y)$ illustrates the effectiveness of the algorithm. In order to obtain a tight estimate, we consider the highest possible value for $U_0^*(y)$ after $100,000$ iterations of FP by all the players. The highest occurs when the first strategy profile is $(C, D, L)$ and game converges to the mixed strategies $(0.45, 0, 0.55)^T$, $(0.35, 0, 0.65)^T$ and $(0, 0.2, 0.8)^T$ generating an expected payoff of $3.89$ for the \textit{IP}. 

Algorithm \ref{algo:compute_X} begins with the initialization of the list $\mathbf{V}[r]$ and the matrix $\mathbf{W}[r,r]$ (Step 1), where $r$ is the cardinality of the set $Y_0$. This is followed by determination of $r$ Linear Programming problems and saving the solution vectors and optimized function values in $\mathbf{W}$ and $\mathbf{V}$ respectively (Steps 4 and 5). Step 7 identifies the highest entry in $\mathbf{V}$ while Step 8 returns the corresponding pure action as $y_0^*$ and the mixed strategy as $\mathbf{z}_0^*$. In Example \ref{eg:play_3_strat_3_advanced}, $y_0^*$ and $\mathbf{z}_0^*$ turn out to be $A$ and $(0, \frac{1}{7}, \frac{6}{7})$ while the corresponding highest payoff for the \textit{IP} is $8.57~(=\frac{60}{7})$. This is considerably higher than $U_0^*(y) = 3.89$ indicating the efficacy of Algorithm \ref{algo:compute_X}. 

Computation of $\mathbf{z}_0^*$ and $y_0^*$ provides the optimal convergence based mixed strategy. The remaining steps in Algorithm \ref{algo:compute_X} calculate a strategy trajectory and begin with identifying the sizes of $\tau^*$ and $\tau'$. In Example \ref{eg:play_3_strat_3_advanced}, the smallest integer that can enforce the probabilities $\frac{1}{7}$ and $\frac{6}{7}$ is $7$, while the FP requires at most $4$ steps to converge to $(A,U,R)$ when the \textit{IP} plays $A$ repeatedly. Step 9 thus computes $\tau^*$ as $7$, Step 10 determines $\tau_0$ as $4$ while Step 11 computes $\tau'$ to be $\max \{4,7\} = 7$. The sequence $\mathbf{X}'$ is obtained via Step 12 as $(A, A, A, A, A, A, A)$. 

Example \ref{eg:play_3_strat_3_advanced} illustrates a scenario where $q^*_{y_0^*} = q_A^* = 0$. Thus, the \textit{IP} need not repeat $A$ to restrict the opponents to the Nash equilibrium of $\mathbf{G}^{(A)}$, once the opponents use FP to converge. 
Steps 13-17 can be skipped for this example and the sequence $\mathbf{X}^*$ begins with $B$, the lowest indexed strategy from $Y_0^*$. Each of the remaining strategies (say $q_s^*$) are repeated $\tau^*q_s^*$ times until all the strategies are exhausted (Steps 18-22). Finally, the sequence $\mathbf{X}^* = (B, C, C, C, C, C, C)$ is returned by Algorithm \ref{algo:compute_X}. The payoff corresponding to $\mathbf{X}^*$ is 
$\frac{60}{7}$ which is indeed the expected payoff 
as $\tau \rightarrow \infty$. 

\end{example}

\section{Conclusions}
\label{sec:conclusions}
Fictitious play is a popular learning algorithm that converges to a Nash equilibrium in many classes of games. Here, we assumed that one player is {\em intelligent} that has access to the entire payoff matrix for the game and need not conform to fictitious play. We show such a player can achieve a better payoff than the one at the Nash Equilibrium. This result can be viewed both as a fragility of the fictitious play algorithm to a strategic intelligent player and an indication that players should not throw away additional information they may have, as suggested by classical fictitious play. Future work will consist of consideration of other learning algorithms and presence of multiple intelligent players.

\balance
\bibliographystyle{IEEEtran}


\end{document}